\documentclass[12pt,formal]{article}

\usepackage{etex}

\usepackage{algorithmic}
\usepackage{algorithm}

\usepackage{fullwidth}
\usepackage{etoolbox}
\usepackage{ifthen}
\usepackage{float}
\usepackage{tikz}
\usepackage{multicol}
\usepackage{todonotes}
\usepackage{enumitem}
\usepackage{wrapfig}
\usepackage{adjustbox}
\usepackage{tkz-graph}
\usetikzlibrary{arrows,automata}
\usetikzlibrary{arrows,shapes,positioning,scopes}
\usetikzlibrary{tikzmark,decorations.pathreplacing,calc,fit}
\usetikzlibrary{arrows,automata}
\usetikzlibrary{arrows,petri}
\GraphInit[vstyle = Shade]
\usepackage{ulem}
\usepackage[latin1]{inputenc}
\usepackage{ifthen}
\usepackage{mathrsfs}
\usepackage{multirow}
\usepackage{rotating}
\usetikzlibrary{matrix,arrows}
\usepackage{bussproofs}
\usepackage{qtree}
\usepackage{easyReview}
\usepackage[polutonikogreek,english]{babel}
\usepackage{tikz}
\usepackage{adjustbox}

\usetikzlibrary{arrows,petri}
\usetikzlibrary{arrows.meta}

\usepackage{pgf}
\usepackage{times}
\usepackage[T1]{fontenc}
\usepackage{amsfonts,amsmath,amssymb,amsthm}
\usepackage[all]{xy}
\let\op=\operatorname

\usepackage{dingbat}
\usepackage{subcaption}
\usepackage{cmll}
\floatstyle{plain}
\floatstyle{ruled}
\restylefloat{algorithm}
\newfloat{myalgo}{tbhp}{mya}
\newcommand{\INDSTATE}[1][1]{\State\hspace{1cm}}
\newlength\myindent
\newtheorem{observation}{Observation}
\newtheorem{fact}{Fact}
\newtheorem{formulation}{Formulation}

\newcommand{\NP}{\ensuremath{\mathbf{NP}}}
\newcommand{\PSPACE}{\ensuremath{\mathbf{PSPACE}}}

\newcommand{\size}[1]{\ensuremath{\left|#1\right|}}

\newcommand{\imply}{\supset}

\newcommand{\mil}{\ensuremath{\mathbf{M}_{\imply}}}

\newboolean{numberspec}
\setboolean{numberspec}{false}
\newcounter{specline}

\newcommand{\spec}[3][{}]{%
   \setcounter{specline}{0}%
   \ifthenelse{\equal{#1}{numbers}}%
      {\setboolean{numberspec}{true}}%
      {\setboolean{numberspec}{false}}%
   \ensuremath{
      \begin{array}{l}
         \ifthenelse{\equal{#2}{}}
            {\numberedline}
            {#2\nl}
         #3
      \end{array}
   }
}

\def\P{\op{\mathbf{P}}}
\def\NP{\op{\mathbf{NP}}}
\def\CoNP{\op{\mathbf{CoNP}}}
\def\PSPACE{\op{\mathbf{PSPACE}}}
{\begin{myalgo}[#1]
\centering
\begin{minipage}{#2}
\begin{algorithm}[H]}%
{\end{algorithm}
\end{minipage}
\end{myalgo}}

\newbool{formal}
\booltrue{formal}

\input{tcilatex}

\begin{document}
  \title{Yet another argument in favour of $\NP=\CoNP$}
  \author{Edward Hermann Haeusler$^\dagger$
           \\ Departmento de Inform\'{a}tica\\ PUC-Rio\\
              Rio de Janeiro, Brasil \\
              $^\dagger$Email:hermann@inf.puc-rio.br}

  \maketitle

   \begin{abstract}
     This article shows yet another proof of $\NP=\CoNP$. In a previous article we proved that $\NP=\PSPACE$ and from it we can conclude that $\NP=\CoNP$ immediatly. The former proof shows how to obtain polynomial and, polynomial in time ckeckable Dag-like proofs for all purely implicational Minimal logic tautologies. From the fact that Minimal implicational logic is $\PSPACE$-complete we get the proof that $\NP=\PSPACE$.  
  This first proof of $\NP=\CoNP$ uses Hudelmaier linear upper-bound on the height of Sequente Calculus minimal implicational logic proofs. In an addendum to the proof of $\NP=\PSPACE$ we observe that we do not need to use Hudelmaier upper-bound, since any proof of non-hamiltonicity for any graph is linear upper-bounded. By the $\CoNP$-completeness of non-hamiltonicity we obtain $\NP=\CoNP$ as a corollary of the first proof. In this article we show a third proof of $\CoNP=\NP$, also providing polynomial size and polynomial verifiable certificates that are Dags. They are generated from normal Natural Deduction proofs, linear height upper-bounded too, by removing redundancy, i.e., repeated parts. The existence of repeated parts is consequence of the redundancy theorem for family of super-polynomial proofs in the purely implicational Minimal logic. Its mandatory to read at least two previous articles to get the details of the proof presented here. The article that proves the redundabcy theorem and the article that shows how to remove the repeated parts of a normal Natural Deduction proof to have a polynomial Dag certificate for minimal implicational logic tautologies.

\end{abstract}

   \section{Introduction}

   In \cite{Exponential}, and \cite{SuperPoly}, we discuss the correlation between the size of proofs and how redundant they can be. A proof or logical derivation is redundant whenever it has sub-proofs that are repeated many times inside it. Articles \cite{Exponential} and \cite{SuperPoly} focus on Natural Deduction (ND) proofs in the purely implicational minimal logic $\mil$. This logic is PSPACE-complete. It polynomially simulates proofs in Intuitionistic Logic and full minimal logic.  The fact that $\mil$ has a straightforward language and only two ND rules is worthy of notice. 

   On the other hand, a theorem prover for $\mil$ is as hard to implement as any other propositional logic. We consider a polynomial-time theorem prover to be efficient, following Cook conjecture\footnote{Cook conjectured that a natural problem has a feasible algorithm iff it has polytime algorithm, see \cite{Cook2000}}. Due to its $\PSPACE$-completeness,  as long as we obtain an efficient implementation for a $\mil$ prover, we will deliver an efficient prover\footnote{Polynomial in time}  for any propositional logic that satisfies the sub-formula principle, see \cite{Haeusler2014}.  The efficiency of a Theorem prover is mainly related to storage usage and processing time. In this article, we focus on storage usage.  The storage use is proportional to the size of the proof generated by the theorem prover. A fairer option is to take the smallest proof of the proving theorem into account, indeed.  We wonder whether for every $\mil$ theorem the shortest proofs are polynomial sized.  We use the term ``short proof''  to denote any proof with the size bounded by o polynomial on their conclusion's size. Since there can be more than one proof with the same size, there also can be more than one shortest proof for a given formula.  This article uses a technique for obtaining short certificates for linear height-bound $\mil$ proofs to provide yet another proof of $\CoNP=\NP$.

   In section~\ref{sec:Background}, we discuss the formulation of the conjecture $\NP\stackrel{?}{=}\CoNP$ in terms of proof systems. The background theory and terminology come from \cite{HaeuslerHugeSmall}, \cite{Addendum}, \cite{Exponential}, \cite{SuperPoly}, \cite{BoSL} and \cite{Studia}.
   In section~\ref{relatedWork}, we briefly review the related work. Section~\ref{ProofNPCoNP} shows the sketch of the proof of $\NP=\CoNP$. Section~\ref{sec:Conclusion} discusses briefly the consequences of this proof and the relationship with the other proofs we know on this conjecture.

\section{Related work}\label{relatedWork}

Theoretically, if for every $\mil$ tautology there is at least one short proof then we have $\NP=\PSPACE$. The proof of a valid formula is a certificate for its validity. Of course, the proof belongs to some proof system. A more precise statement has to consider the proof system. In \cite{Krajicek}, Theorem 4.1.2 provides a more precise statement on the relationship between the existence of short proofs for all Classical tautologies in some (Classical) proposition proof system. In section~\ref{sec:Formulation} we discuss a bit more the statements related to $\NP=\PSPACE$ and $\NP=\CoNP$. We also make precise the statement we consider for proving that $\NP=\CoNP$. 

In \cite{BoSL} and \cite{Studia}, we show the existence of short certificates for every $\mil$ valid formula via compression of Natural Deduction proofs into Directed Acyclic Graphs (DAGs). We are aware that $\NP=\PSPACE$ implies that $\NP=\CoNP$. Thus, in \cite{BoSL} and \cite{Studia} we also have proved that $\NP=\CoNP$. 
Compressing proofs in $\mil$ can provide very good glues to compress proofs in any logic satisfying the sub-formula principle.  The Classical Propositional Logic is among these logics. In \cite{BoSL} and \cite{Studia}, we prove that for every \mil tautology $\alpha$ there is a two-fold certificate for the validity of $\alpha$ in \mil. The certificate is polynomially sized on the length of $\alpha$ and verifiable in polynomial time on this length too. The general approach, described in \cite{Studia} and \cite{BoSL},  to prove $\NP=\PSPACE$ obtains short certificates by compressing  Natural Deduction (ND) proofs into  DAGs that eliminate the repetition of redundant parts in the original ND proofs. It is well-known that $\NP=\PSPACE$ implies $\CoNP=\NP$ so we have $\CoNP=\NP$. Using articles \cite{Exponential}, \cite{SuperPoly}, and one of the appendixes of \cite{Exponential} or \cite{Addendum},  we can provide an alternative and more intuitive proof that $\NP=\CoNP$. Moreover, In \cite{Addendum},  we discuss a simpler proof of $\NP=\CoNP$, this time with a double certificate on linear height normal proofs, without to use Hudelmaier result, that is essential in \cite{BoSL}. Thus, this article provides a direct and more intuitive proof of $\NP=\CoNP$ using a method slightly different from what we use in \cite{BoSL} and \cite{Studia}.  

In \cite{HaeuslerHugeSmall}, we show how to use the inherent redundancy for huge proofs, see \cite{SuperPoly} for the redundancy theorem,  to have polynomial and polytime certificates by the removal of all redundant parts of the proofs. Indeed, we collapse all the redundant sub-proofs into only one occurrence. We start with tree-like Natural Deduction proofs and end up with a labelled r-DAG (rooted Directed Acyclic Graph). In section~\ref{ProofNPCoNP}, we explain what we said in the last two phrases in more detail. The use of the redundancy theorem and corollary shown in \cite{SuperPoly},  the essence of this proof's approach, does not seem to be easily adaptable to proof of $\NP=\PSPACE$, indeed. In \cite{BoSL} the linearly height upper-bounded proofs of the tautologies in \mil do not need to be normal proofs. 

In 2015, we divulgated the first version of our proof that $\NP=\PSPACE$ in the Foundations of Mathematics [FOM] forum on the internet. Many questions arose on the known exponential lower-bounds for propositional proofs in Classical Logic in Frege systems. One of the main sources of this research and the results are Reckhow thesis and Hrubes article. We report here why these lower-bounds do not disturb our result. There are two main questions. The first is that one can see Natural Deduction as a Frege system and hence the exponential lower-bounds could apply to them too. The second tries to relate the known speed-up of Lemon style Natural Deduction where proofs are lists, and the repetition of the hypothesis is not need, with regard to Gentzen/Prawitz proof-trees where a formula can have many occurrences as a hypothesis in a proof. Below we answer to both questions. Professors Richard Zach, Joe Shipman and Thimothy Chown raised most of the questions. All the discussion, including the answers below, are registered in the FOM records between August 2015
and September 2016.

The lower-bounds obtained in Reckow's thesis \cite{Reckhow} as
well in Hrubes paper, \cite{Hrubes}, is strongly based on the fact that an axiomatic
system is used and each line in the proof is a formula, opposed to set of
formulas.   The understanding that dag-like proof can be obtained  from a
Frege system by allowing  the use of a formula  more than once as a premise
is correct but does not apply directly to what we are doing and to
Lemmon style ND too.  One of the main features of ND is the discharging of
hypothesis, a device that Frege systems obtain by means of the Deduction
Theorem (a Meta-Theorem). The implication-introduction rule is a powerful
device to shorten proofs since we can have with its use the formation of
lemmas. Reckow's thesis, \cite{Reckhow}, discusses Natural Deduction in pages 64-68.  Below
we can read an excerpt of page 64, where  "this theorem"  is the Deduction
Theorem.
\begin{quotation}
The important thing to note about this theorem, however, is that it provides
a new kind of inference rule. This rule· allows one to infer something
about the derivability of a certain formula from a certain set of formulas by
showing
a (possibly easier) derivation of a slightly different formula from a slightly
different set of hypotheses. Thus, there is the potential, at least, for much
shorter derivations by appealing to this rule.
\end{quotation}

The solution Reckow provides, afterwards,  for incorporating Natural
Deduction to Frege-like systems was to consider each line, in a more
general Frege system approach, as a pair  $S\vdash A$, where $S$ is a set of
formulas and $A$ is a formula. In this way, the Deduction Theorem is turned
into an inference rule, infer  $S\vdash B \rightarrow A$ from $S,B\vdash A$.
But now, we have a problem with ensuring compression. There are
exponentially  many subsets of  sub formulas of the conclusion (the
sub-formula property).  The combinatorics changes against compression of
proofs if we consider this way of representing ND systems. There are proofs
of some formulas that have more than polynomially many pairs of the above
kind, we cannot ensure any more that a super-polynomially sized graph is
labelled with polynomially many labels (the lines), and hence there is at
least one  label (formula) that occurs polynomially many times in the
graph. This is essential in our approach that there are only polynomially
many  possible labels  (in fact there are only linearly many sub formulas
of the conclusion of the proof). The compression of the proof is provided
by collapsing these formulas that occur super-polynomially many times, and
they exist because of the super-polynomial gap between the number of labels
of the vertexes and the size of the graph itself.
Finally, if we consider Lemmon Style Natural Deduction, each line is a
formula plus the annotation indicating the premises (number line) used to
derive this formula.  As discussed above shows,  this system is more economical
regarded the (natural) solution proposed by Reckow thesis, that is in fact
quite close to those proof representations implemented by Gentzen's Sequent
Systems and Tableaux.

On the other hand, the use of Lemmon Style ND does not allow us to conclude
that the existent lower bound for Frege systems directly apply to ND.   A
dag-like proof, even in Lemmon Style,  cannot simply indicate that we are
using a formula more than once. Each formula has it is own natural
dependency set of formulas. We cannot collapse equal formulas derived  from
a different set of formulas without taking care of annotating the new
dependency, and if we take this control seriously, we end up in a dag-like
proof as defined in our article, and having to adjust the annotations to
polynomial size and so on.

We conclude that the exponential lower-bounds are for Frege systems. Because of the features analysed above, Frege systems do not polynomially simulate Dag-proofs that are the main representation used in our proof systems able to obtain short validity certificates.


\section{On the formulation of the conjecture $\NP\stackrel{?}{=}\CoNP$ in terms of proofs}\label{sec:Background}\label{sec:Formulation}

  An alphabet is any non-empty and finite set of symbols. Given an alphabet, $\Sigma$, the set of all strings, including the empty string ($\epsilon$), is denoted by $\Sigma^{*}$. A formal language on the alphabet $\Sigma$ is any set $L\subseteq\Sigma^{*}$.

With the sake of a faster and more efficient presentation we take for free the definition of a (propositional) logic $L$ as a pair $\langle For_{L},\models_{L}\rangle$. $For_{L}\subseteq\Sigma_{L}^*$ is the set of formulas of $L$, where $\Sigma_{L}$ is the alphabet of the logic. $\models_{L}\subseteq 2^{For_{L}}\times For_{L}$ is the logical consequence relation that defines when given a set $\Delta\in 2^{For_L}$ and a formula $\alpha\in For_{L}$ if $\alpha$ is logical consequence of $\Delta$, $\Delta\models_{L}\alpha$ in symbols. We use $\models_{L}\alpha$ as a shorthand for $\emptyset\models_{L}\alpha$.
The set of valid formulas, or tautologies, is $Taut_{L}=\{\alpha:\mbox{$\models_{L}\alpha$ and $\alpha\in For_{L}$}\}$. A proof system for $L$ is a set $Prov_{L}\subseteq (\Sigma_{Prov}\cup\Sigma_{L})^*$, where $\Sigma_{Prov}$ is the proofs' alphabet, together with a function  $conc$ that maps members of $Prov_{L}$, the proofs, to formulas in $For_{L}$, their respective conclusions. Given a proof system $P$ for $L$, we define $\vdash_{Prov_{L}}\alpha$, iff, there is $\pi\in Prov_{L}$, such that, $conc(\pi)=\alpha$. $Prov_{L}$ is sound and complete, iff, for all $\alpha$, $\vdash_{L}\alpha$ if and only if $\models_{L}\alpha$. The size of a formula and a proof are the length or size of the respective strings. We use $\size{s}$ for denoting the size or  length of $s$,i.e., the number of occurrences of symbols in $s$,  a string. For any logic $L$, we require that $Prov_{L}$ is decidable and $conc$ is computable. We require more, the existence of polytime algorithms implementing each of them. We know that Classical, Intuitionistic and full Minimal Logic and \mil have such polytime algorithms, indeed. They form the scope of logics in this article.    

     In this section, we discuss the mathematical formulation of the conjecture $\NP\stackrel{?}{=}\CoNP$ regarding the size of the certificates for membership. Our definition of a proof system follows \cite{Krajicek}, but inside the scope of General Logics\cite{Meseguer, JLC}. Theorem 4.1.2, page 24 in \cite{Krajicek} provides a precise statement on the relationship between the existence of short proofs for all Classical tautologies in some (Classical) proposition proof system. We examine this statement and provide an equivalent formulation that is more suitable for our purposes. In the sequel, we state the primary definition of some complexity classes used in this article.  

 Due to technical facts, when analysing the complexity of algorithms we consider only languages on alphabets with at least two symbols\footnote{This is a consequence of a result, see \cite{Arora},  that says that if there is an $\NP$-complete unary language then $\NP=\P$}. We take for free the use of the  $\mathcal{O}$ notation.  Given a formal language $L\subseteq\Sigma^{*}$, we say that:
   \begin{enumerate}
   \item $L\in \P$ whenever there is an algorithm $\chi_{L}$, such that, for all $\omega\in\Sigma_{L}$:
     \[
     \chi_L(\omega)=\left\{\begin{array}{ll} 1 & \mbox{if $\omega\in L$} \\
     0 & \mbox{if $\omega\not\in L$} \end{array}\right.
     \]
     Moreover, $Steps(\chi_{L})\in\mathcal{O}(n^q)$, for some $q\in\mathbb{N}$
   \item $L\in \NP$\label{positive} whenever there is an algorithm $\varphi_{L}$, such that, for all $\omega\in\Sigma^{*}$, there is  $c_{\omega}\in\Sigma_c^{*}$, $\Sigma\subseteq\Sigma_c$, $\size{c_{\omega}}\leq \size{\omega}^q$, $1\leq q$ and we have that: $\omega\in L$ if and only if  $\varphi_{L}(\omega,c_{\omega})=1$  and $Steps(\varphi_{L})\in\mathcal{O}(n^p)$, $1\leq p$.
   \item $L\in \CoNP$\label{negative} whenever there is an algorithm $\varphi_{L}$, such that, for all $\omega\in\Sigma^{*}$, there is  $c_{\omega}\in\Sigma_c^{*}$, $\Sigma\subseteq\Sigma_c$, $\size{c_{\omega}}\leq \size{\omega}^q$, $1\leq q$ and we have that: $\omega\not\in L$ if and only if $\varphi_{L}(\omega,c_{\omega})=1$  and $Steps(\varphi_{L})\in\mathcal{O}(n^p)$, $1\leq p$. 
   \item $L\in \PSPACE$  whenever there is an algorithm $\chi_{L}$, such that, for all $\omega\in\Sigma_{L}$:
     \[
     \chi_L(\omega)=\left\{\begin{array}{ll} 1 & \mbox{if $\omega\in L$} \\
     0 & \mbox{if $\omega\not\in L$} \end{array}\right.
     \]
     Moreover, $Cells(\chi_{L})\in\mathcal{O}(n^q)$, for some $q\in\mathbb{N}$
   \end{enumerate}

The functions $Steps$ and $Cells$ compute the number of steps and memory usage of the algorithms that are their respective arguments. Their names come from its original terminology in terms of Turing Machines. 

   In the item~\ref{positive} we say that  $c_{\omega}$ is a  positive certificate for membership of $\omega$ to $L$, and in item~\ref{negative}, $c_{\omega}$ is a negative certificate for membership of $\omega$ to $L$, i.e., positive certificate for membership to $\overline{L}$\footnote{Given a formal language $L$over $\Sigma$, $\overline{L}=\Sigma^*-L$}. They are considered as easy certificates since their verification is polynomial in time and their sizes (lenght) is polynomial on $\omega$ too.

   Roughly speaking, $\NP$ is the set of formal languages that have easy, so to say, polynomially sized and polytime verifiable positive certificates for each language element. A language $L$, over $\Sigma$, is in $\CoNP$  if and only if it has an easy negative certificate to each  $\omega\in\Sigma^*$ such that $\omega\not\in L$. Due to Fact~\ref{Fact:InvFinito}, we can say that for all elements of a language $L\in \CoNP$, but finitely many, if they have easy positive certificates then $L\in \NP$.

   We have the following fact.

   \begin{fact}\label{Fact:InvFinito} Let $L$ and $L^{\prime}$ be formal languages, such that, $(L-L^{\prime})\cup(L^{\prime}-L)$ is finite. So $L\in \mathcal{C}$ iff $L^{\prime}\in\mathcal{C}$, for $\mathcal{C}$ = $\NP$, $\CoNP$, $\PSPACE$ and any other natural complexity class.
\end{fact}

 \begin{observation}\label{Obs:WLOG}  A certificate is hard if and only if it is not easy. Without loss of generality, hard certificates have super-polynomial sizes. This fact is a consequence of lemma~\ref{Lemma:HardCertificates}.
 \end{observation}
 
   \begin{lemma}\label{Lemma:HardCertificates}
     Let  $L\in \NP$,  respectively $L\in \CoNP$,  be such that its elements, but finitely many,  have polynomial positive, respectively negative,  certificates.  Moreover, the best algorithm, $\mathcal{O}$, to verify the validity of the certificates, is not polynomially upper-bounded on the size of the certificates. There is a polynomial-time algorithm $\mathcal{A}$, such that,  for all elements of $L$, but finitely many, there are positive, respectively negative, super-polynomial size certificates can be checked in polynomial time by $\mathcal{A}$.
   \end{lemma}

   \begin{proof}

   We consider, for each element of $e\in L$ the formal execution of $\mathcal{O}$ on its corresponding  polynomial certificate $c_e$. This execution, or better, the trace of $\mathcal{O}(c_e)$, namely $Tr(\mathcal{O}(c_e)$, is the certificate for the element $e$.  We can observe that any polynomial does not bound $Tr(\mathcal{O}(c_e))$ by hypothesis. Of course, the verification of a valid universal Turing machine execution is an algorithm that runs in linear, hence polynomial,  time on its trace. $Tr(\mathcal{O}(c_e))$  checking runs by checking that each step in the trace is valid.
   \end{proof}

   \begin{fact}\label{Fact:CoNPsubNP}
     If $\CoNP\subseteq\NP$ then $\CoNP=\NP$.
   \end{fact}

   \begin{proof}
        \[
  \begin{array}{c}
    \mbox{Assume} \CoNP\subset \NP \\
    \Downarrow \\
\mbox{$L\in \NP$, iff, $\overline{L}\in \CoNP$, $\CoNP\subseteq  \NP$ so $\overline{L}\in \NP$, iff, $L\in \CoNP$}\\
    \Downarrow \\
CoNP=NP
  \end{array}
  \]
  \end{proof}

As already mentioned, if we prove that a language $L\in \CoNP$ has easy certificates for all, but finitely many, elements then we have that  $L\in\NP$. Moreover,  if the language $L$ is $\CoNP$-complete then we have $\CoNP\subseteq \NP$ 
By fact~\ref{Fact:CoNPsubNP} we would have that $\CoNP=\NP$. We remember that $Taut_{Cla}$ is $\CoNP$ complete. 
Thus, $\CoNP=\NP$ if and only if, for all, but finitely many, classical tautologies,  they do have easy (positive) certificates. Equivalently, by  observation~\ref{Obs:WLOG},  $CoNP\neq\NP$ if and only if for all, but finitely many, formulas $\varphi_{e}\in Taut_{Cla}$, the smallest certificate for $\varphi_{e}$ is not upper-bounded by any polynomial on the size of $\varphi_{e}$. Thus, we have a possible formulation of $\NP\neq\CoNP$.

\begin{formulation}\label{For:First}
$\NP\neq \CoNP$  iff for each proof system $\mathcal{P}$ for the Classical Logic, for all $\pi\in\mathcal{P}$, but finitely many, $\pi$ is not upper-bounded by any polynomial on the size of $conc(\pi)$.  
\end{formulation}

The formulation above is equivalent to the following. We first define the function $F_{\mathcal{P}}$, given a sound and complete proof system $\mathcal{P}$ for the classical logic.
\[
F_{\mathcal{P}}(n)=Min(\{\size{\pi}:\mbox{$conc(\pi)=\alpha$, $\pi\in Prov_{\mathcal{P}}$ and $n=\size{\alpha}$}\})
\]

Thus, we reach our formulation of the statement for $\NP\neq\CoNP$.

\begin{formulation}\label{For:NPneqCoNP}  $\NP\neq\CoNP$ if and only if for every proof system $\mathcal{P}$ for the Classical Logic, $F_{\mathcal{P}}$ is a super-polynomial function from $\mathbb{N}$ into $\mathbb{N}$
  \end{formulation}

We can refine the above Formulation~\ref{For:NPneqCoNP} a bit more.  In \cite{Arora}, we can find a  well-known reduction, used in \cite{Addendum} and appendix of \cite{Exponential} that for each graph $G$ assigns a propositional formula $\alpha_{G}$, where $\alpha_G$ is satisfiable iff $G$ is Hamiltonian. In \cite{Addendum} and \cite{Exponential}, appendix in pages 17--22, we first observe that $\alpha_G$ is unsatisfiable if $G$ is not Hamiltonian. Due to Glyvenko theorem\footnote{Glyvenko theorem is: For any propositional formula $\varphi$, $\vdash_{Cla}\neg\varphi$ if and only if $\vdash_{Int}\neg\varphi$}, we can state that $\neg\alpha_G$ is an intuitionistic tautology iff $G$ is not Hamiltonian. Moreover, we show how to define the formula $\beta_G$, deductively equivalent to $\neg\alpha_G$ in minimal logic. $\beta_G$ uses only the implication connective and a new propositional letter $q$, that plays the role of the absurdity logical constant. $\beta_G$ has polynomial length when compared to the number of vertexes of the graph $G$. The height of the proofs of $\beta_G$, whenever $G$ is non-hamiltonian, is linear on the number of vertexes of $G$, and it is a normal proof in \mil too.    
Let $NHam$ be the set of all proofs of non-hamiloniticity for all non-hamiltonian graphs $G$. Hence, $NHam$ is the set of all proofs of $\beta_G$ for non-hamiltonian graphs $G$. Without loss of generality, we only consider normal proofs with height linear bounded. We have then the following proposition used in this article and in \cite{SuperPoly}  is a consequence of Formulation~\ref{For:NPneqCoNP}.

\begin{proposition}\label{Prop:NPneqCoNP}
  Consider the set $NHam$ of the all linear height-bounded normal proofs in \mil of the formulas $\beta_G$, for $G$ non-hamiltonian.
  If $\NP\neq\CoNP$ then the function:
  \[
  F_{NHam}(n)=Min(\{\size{\pi}:\mbox{$\pi$ is a proof of $\beta_G$ and $n=\size{\beta_G}$}\})
  \]
  is super-polynomial
\end{proposition}

In \cite{SuperPoly}, proposition 19 in appendix B, page 22, can be applied to conclude that $F_{NHam}$ is super-polynomial if and only if $NHam$ is an unlimited set of super-polynomial proofs. From $\NP\neq\CoNP$ we conclude that $NHam$ is a family of super-polynomial proofs. See \cite{SuperPoly} to discuss the concept of a family of super-polynomial proofs and family of proofs that have super-polynomial lower-bounds. With the sake of facilitating to read  this article,  we write down below the definition of a set of super-polynomial proofs and the proposition 19, here state and contextualized for this presentation as proposition~\ref{prop:EquivSuperSuper} 

\begin{definition}\label{def:unlimited}
  A set $\mathcal{S}$ of Natural Deduction proof-trees is unlimited, if and only if, for every $n>0$ there is $\Pi\in \mathcal{S}$, such that, $\size{\Pi}>n$.
\end{definition}

In the following definition, $c(\Pi)$ denotes the formula that is the conclusion of $\Pi$.

\begin{definition}\label{def:FuncaoF}
  Let $\mathcal{S}$ be an unlimited set of N.D. proof-trees. Let  $\mathcal{S}_{m}(x)$ be the predicate $(x\in\mathcal{S}\;\land\;\size{c(x)}= m)$, for $0<m\in\mathbb{N}$. We define the function $F_{\mathcal{S}}:\mathbb{N}\longrightarrow\mathbb{N}$ that associates do each natural number $m$ the size of one of the least N.D. proof-tree $x$ satisfying $\mathcal{S}_{m}(x)$.
  \[
  F_{\mathcal{S}}(m)=\left\{\begin{array}{ll} 0 & \mbox{if $m=0$} \\
                                      Min(\{\size{x}:\mathcal{S}_{m}(x)\}) & \mbox{if $m>0$}               
  \end{array}\right.
  \]
\end{definition}

\begin{definition}\label{def:Superpoly} Let $\mathcal{S}$ be a set of Natural Deduction proofs, such that:
  $\Pi\in \mathcal{S}$ if and only if  ($\forall p\in\mathbb{N}$, $p>0$, $\exists n_0$,$\forall n>n_0$, $\size{c(\Pi)}=n$ and $\size{\Pi}>n^p$)
 In this case, we say that $\mathcal{S}$ is a set of super-polynomially sized ND proofs.
\end{definition}

 The following proposition explains why the name we used in Definition~\ref{def:Superpoly}

 \begin{proposition}\label{prop:EquivSuperSuper}
   Let $\mathcal{S}$ be an unlimited set of ND proof-trees. We have that $\mathcal{S}$ is a set of super-polynomially sized proofs if and only if
   $F_{\mathcal{S}}$ is a super-polynomial function from $\mathbb{N}$ into $\mathbb{N}$.
   \end{proposition}

We provide in the next section a precise argument for $\NP=\CoNP$. This proof's strategy is use Formulation~\ref{For:NPneqCoNP}  and Proposition~\ref{prop:EquivSuperSuper} to conclude that $NHam$ is a set of super-polynomially sized proofs, from the hypothesis that $\CoNP\neq\NP$. From the fact that $NHam$ is a set of super-polynomial proofs, according to Definition~\ref{def:Superpoly}, we reach a contradiction, concluding that $\NP=\CoNP$. 

\section{A proof for $CoNP=NP$}\label{ProofNPCoNP}

In the previous section, we show that when considering the complexity class $\CoNP$, we can only consider linearly height-bounded normal proofs in \mil. The proofs, in \mil,  of the non-hamiltonianicity of graphs, are linearly height bounded. See the appendix in ~\cite{Exponential} or ~\cite{Addendum} for a detailed explanation on this. If $NP\neq CoNP$ then the set of non-hamiltonian graphs has no polynomially sized and verifiable in polynomial time certificate for each of its elements. Hence, by assuming that $NP\neq CoNP$, we have to conclude that $NHam$ is a family of normal super-polynomial proofs with linear height, see proposition~\ref{Prop:NPneqCoNP} in the last section. If we consider any proof in $NHam$, either it is polynomially sized, and we have nothing to prove, or it is bigger than 
 $m^p$, for some $p>3$, where $m$ is the size of the proof's conclusion. We observe that the case $p\leq 3$ is subsumed by $p>3$, anyway.  We can apply  Theorem of redundancy, theorem 14 in \cite{SuperPoly}, to show that this big proof is redundant so that we can apply the compression algorithm, algorithm 2 in \cite{HaeuslerHugeSmall}, to obtain a correct rDagProof of size smaller than $m^p$, according to Lemma 15 in \cite{HaeuslerHugeSmall}. Finally, Algorithm 3 in \cite{HaeuslerHugeSmall}, page 32,  can check the correctness of this polynomially sized rDagProof in time upper-bounded by $m^{4p}$.

 We provide a proof that shows a polynomial certificate for each non-hamiltonicity of each non-hamiltonian graph. We can check that each of these polynomial certificates is a (correct) certificate in polynomial time too, applying algorithm 3 in \cite{HaeuslerHugeSmall}. We can conclude that $CoNP\subseteq NP$, since non-hamiltonicity of graphs is a $CoNP$-complete problem. 

 \begin{theorem}
   $\NP=\CoNP$
   \end{theorem}

\section{Conclusion}\label{sec:Conclusion}

This article provides yet another proof of $NP=CoNP$. In \cite{BoSL} we have a proof that $NP=NPSPACE$. An immediate consequence of this equality is that $NP=CoNP$. The approach that arises from the results we have shown here does not need Hudelmaier~\cite{Hudelmaier} linearly bounded sequent calculus for \mil logic. The proof reported in~\cite{BoSL}, on the other hand, needs Hudelmaier Sequent Calculus and a translation to Natural Deduction proofs that preserves the linear upper-bound. However, the resulted translation is not normal, and it is well-known that normalization does not preserve upper-bounds in general. Thus, we cannot apply our approach to the whole class of \mil tautologies to prove that $NPSPACE\subseteq NP$, for the use of normal proofs is essential to obtain the redundancy Theorem, i.e., Theorem 14. The compression method reported in \cite{HaeuslerHugeSmall}, uses the redundancy Theorem that is essential to prove \mil short tautologies automatically. It seems easier than the use of the double certificate approach in~\cite{BoSL}.

We observe that the proof is by contradiction. We start by assuming that $\NP\neq\CoNP$, and, using Formulation~\ref{For:NPneqCoNP} we have the existence of $NHam$ as a family of super-polynomial proofs that are linear height-bounded and normal.  Using the redundancy Theorem we obtain, by compression, a polynomial certificate, a Dag Proof, that is polynomial-time verifiable too, having a contradiction.

We show that for any huge proof of a tautology in \mil we obtain a succinct certificate for its validity. Moreover,  we offer an algorithm to check this validity in polynomial time on the certificate's size. We can use this result to provide a compression method to propositional proofs. Moreover, we can efficiently check the compressed proof without uncompressing it. Thus, we have many advantages over traditional compression methods based on strings. The compression ratio of techniques based on collapsing redundancies seems to be bigger, as shown in \cite{FlavioHaeuslerEBL} that reports some experiments with a variation of the Horizontal Compression method compared with Huffman compression. The second and more important advantage is the possibility to check for the validity of the compressed proof without having to uncompress it. In general, the original proof is huge, super-polynomial and hard to check computationally.

A last, technical observation, concerns the fact that the proof, although less intuitive, can also proceed by using the set of all classical tautologies, instead of $NHam$. Consider a classical tautology $\alpha$. By the $\CoNP$-completeness of the set of classical tautologies we have the existence of a graph $G_{\alpha}$ that is non-hamiltonian, and $\beta_{G_{\alpha}}$ is a \mil tautology that has normal proofs linear height-bounded on the size of $k=\size{\beta_{G_{\alpha}}}$. Thus, it has a rDag certificate, polynomial on $k$. Thus, without any loss of generality, we can use the set of all \mil normal linear height-bounded proofs instead of $NHam$.

\section{Acknowledgement}

We would like very much to thank professor Lew Gordeev for the work we have done together and the inspiration to follow this alternative approach. Thank Professor Luiz Carlos Pereira for his support, lessons and ideas on Proof Theory since the first course I have taken with him as a student. Thank the proof-theory group at Tuebingen-University, led by prof. Peter Schroeder-Heister. We want to thank Thomas Piecha and M. Arndt. Many thanks to profs Gilles Dowek (INRIA) and Jean-Baptiste Joinet (univ. Lyon) for the intense interaction during this work's elaboration. Finally, we want to thank all students, former students, and colleagues who discussed with us in many stages during this work. We must have forgotten to mention someone, and we hope we can mend this memory failure in a nearer future. Special thanks go for Alex Vasconcelos Garcia, Christian Renteria and Eduardo Laber, they pointed out many imprecision mistakes in previous versions. The doubts raised from a talk I gave in Université de Paris XIII about one of the earlier versions were important to us to have what we think is a relatively more intuitive proof. We would like to thank Leonardo Moura and Christiano Braga for many helpful suggestions on the way we should conduct this research.

\end{document}